\documentclass{article}	

\usepackage{amsmath}
\usepackage{amsthm}
\usepackage[dvips]{graphicx}

\newtheorem{thm}{Theorem}
\newtheorem{lem}{Lemma}
\newtheorem{cor}{Corollary}
\newtheorem{cla}{Claim}
\newtheorem{de}{Definition}

\begin{document}

\title{On the Hairpin Incompletion}

\author{{\bf Fumiya Okubo$^1$} and {\bf Takashi Yokomori$^2$\footnote{Corresponding author}}\\[2mm]
$^1$Graduate School of Education\\
Waseda University, 1-6-1 Nishiwaseda, Shinjuku-ku\\
Tokyo 169-8050, Japan\\
{\tt f.okubo@akane.waseda.jp}\\[2mm]
$^2$Department of Mathematics\\
 Faculty of Education and Integrated Arts and Sciences\\
 Waseda University, 1-6-1 Nishiwaseda, Shinjuku-ku\\
Tokyo 169-8050, Japan\\
{\tt yokomori@waseda.jp}}

\date{}
\maketitle

\begin{abstract}
Hairpin completion and its variant called bounded hairpin 
completion are operations on formal languages, inspired by a hairpin 
formation in molecular biology. Another variant called 
hairpin lengthening has been recently introduced and studied on the closure  
properties and algorithmic problems concerning several families of languages. 
   
In this paper, we introduce a new operation of this kind, called {\it hairpin incompletion} which 
is not only an extension of bounded hairpin 
completion, but also a restricted (bounded) variant of 
hairpin lengthening.  Further, the hairpin incompletion 
operation provides a formal language theoretic framework 
that models a bio-molecular technique nowadays known as  
Whiplash PCR.  We study the closure properties 
of language families under both the operation and its iterated version. 

We show that  a family of languages closed under 
intersection with regular sets, concatenation with regular sets, and 
finite union is closed under one-sided iterated hairpin incompletion, and  
that a family of languages containing all linear languages and 
closed under circular permutation, left derivative and substitution 
is also closed under iterated hairpin incompletion.

\end{abstract}

\section{Introduction}

In these years there has been introduced and intensively investigated  an operation called {\it hairpin completion} in formal language theory, inspired by intra molecular phenomena in molecular biology. A hairpin structure is well-known as one of the most popular secondary structures for a single stranded DNA (or RNA) molecule to form, with the help of so-called {\it Watson-Crick complementarity} and {\it annealing}, under a certain 
biochemical condition in a solution.

This paper continues research directed by a series of works  started in \cite{tc06}  where the hairpin completion operation was introduced, followed by several other related papers (\cite{dam09,cie2007,tcs09}), where both the hairpin completion and its inverse operation (the hairpin reduction) were  investigated. 

Inspired by threefold motivations, we will introduce the notion of {\it hairpin incompletion} in this paper.  Firstly, the hairpin incompletion is a natural extension of the notion of {\it bounded hairpin completion} introduced  and studied in \cite{bounded09} which is a restricted variant of the hairpin completion with the property that the length of the 
prefix (suffix) prolongation is constantly bounded.  Thus, the bounded 
 hairpin completion involves the lengthening of prefix (suffix) with a 
 constant length of the strand at the end, which implies that the 
 resulting strand always bears a specific property that its prefix 
and suffix {\it always form complementary sub-strands} of a certain  
 constant length. In contrast, our notion of hairpin incompletion can produce a resulting strand with more complexity, due to the nature of 
its prolongation, which will be formally explained later. 

Secondly, the hairpin incompletion is also regarded as a restricted 
variant of the notion of {\it hairpin lengthening} recently introduced 
in \cite{hl10} which is an extension of the (original) notion of the hairpin completion. More specifically, the hairpin lengthening  
concerns the prolongation of a strand that allows to stop itself at 
any position in the process of completing a hairpin structure.  From the practical and molecular implementation point of view, here we are interested in the case where the prolongation in the hairpin 
lengthening is bounded by a constant, which leads to our notion of 
the hairpin incompletion. In this respect,  one may take the hairpin incompletion as the {\it bounded} variant of the hairpin lengthening.   

Thirdly, the hairpin incompletion  can provide a purely formal framework that exactly models a bio-molecular technique called 
{\it Whiplash PCR} that has nowadays been recognized as a promising experimental technique and has been proposed in an ingenious paper  \cite{mH00} by Hagiya et al.  They developed an experimental technique  called polymerization stop and theoretically showed in terms of thermal cycling how  DNA molecules can solve the learning problem of $\mu$-formulas (i.e.,  Boolean formulas with each variable appearing only once) 
from given data.  Suppose that a DNA sequence is designed as given 
in (a) of Figure 1, where  a sequence of transition (program) 
is delimited by a special sequence (called {\it stopper sequence}) 
and $\alpha$ and its reversal complementarity $\bar{\alpha}^R$ may hybridize, leading to a hairpin structure (b). Then, the head $\bar{\alpha}^R$ (current state)  is extended by polymerization 
(with a primer $\bar{\alpha}^R$ and a template $\gamma$) 
up to $\bar{\gamma}^R$, where the stopper sequence is specifically 
designed to act as the stopper.  In this way, this cycle can execute one process of state transition and be repeatedly performed\footnote{
Adleman has named this experimental technique whiplash PCR}. 
Following the work of \cite{mH00},  Sakamoto et al.   has shown  
how some NP-complete problems can be solved with Whiplash PCR (or Whiplash machines)  (\cite{Sakamoto99BS}). Recently, Komiya et al. has demonstrated 
the applicability of Whiplash PCR to the experimental validation 
of signal dependent operation (\cite{kK08a}).

\begin{figure}[t]
\centerline{
\includegraphics[scale=0.3]{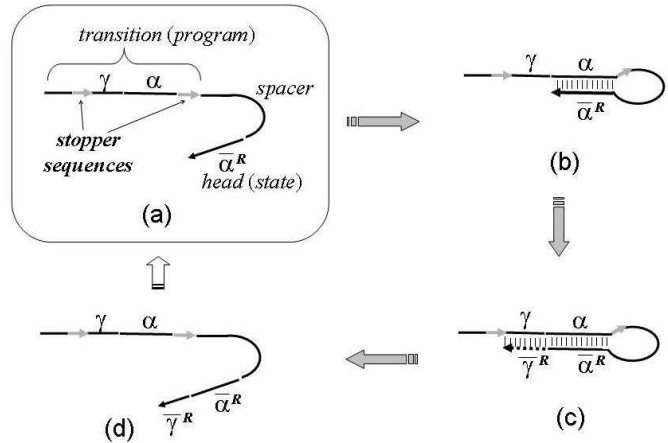}}
\caption{(a)The structural design of Whiplash PCR molecule ;  
(b) hairpin formation with stem part $\alpha$ ;
(c) polymerization extension of $\gamma$ ;
(d) simulation of one state transition.}
\label{hi2}
\end{figure}

The paper is organized as follows.  After providing the definitions of 
the basic concepts used in the paper, we define the central notion 
of {\it hairpin incompletion} (as an extension of the bounded hairpin 
completion and also as a bounded variant of the hairpin lengthening) in Section 2. We first show in Section 3 that any family of languages with a certain closure properties is closed under the hairpin incompletion. 
 We then consider the case of applying the {\it iterated}  hairpin
 incompletion operations, and show that every AFL is closed under the
 iterated one-sided hairpin incompletion. This result is further extended
 to the general case of the iterated hairpin incompletion, 
and it is shown that any family of languages including all linear languages 
and with a certain closure properties is also closed under the iterated 
hairpin incompletion, and as a corollary that the family of context-free 
languages is closed under the iterated hairpin incompletion, followed 
by a brief discussion with concluding remarks in Section 4.

\section{Preliminaries}

\subsection{Basic definitions}

This paper assumes that the reader is familiar with the basic notions of formal language 
theory \cite{handbook}. In particular, for the notions of abstract family of languages,  
we refer to \cite{salomaa73}. 

For an alphabet $V$, $V^*$ is the set of all finite-length strings of symbols from $V$, 
and $\lambda$ is the empty string. while $V^+$ denotes $V^*-\{\lambda\}$.  
For $w \in V^*$, $|w|$ is the length of  $w$. For $k \ge 0$,  we define $V^{\ge k} = 
\{ w \in V^* \, | \, |w| \ge k \}$.  Note that for a set $S$, $|S|$ denotes the cardinality of $S$.  

For $k \ge 0$, let $pref_k (w)$ and $suf_k (w)$ be the prefix and the suffix of a word $w$
 of length $k$, respectively. For $k \ge 0$, we define  $Pref_{\le k} (w) = \{ pref_i (w) \, | \, 0 \le i \le k \}$ 
and $Suf_{\le k} (w) = \{ suf_i (w) \, | \, 0 \le i \le k \}$.  For $k \ge 1$, let $Inf_k (w)$ be 
the set of infixes of $w$ of length $k$. If $|w| \le k-1$, then 
$pref_k (w)$, $suf_k (w)$ and $Inf_k(w)$ are all undefined.  (Note that for $w\in V^+$, $pref_k (w)$ and $suf_k (w)$ are elements 
in $Inf_k(w)$.)  By $w L$ ($L w$) we simply denote $\{w\} L$ ($L \{w\}$), i.e., the concatenation of $w$ with a 
language $L$.  The left derivative of a language $L$ with a word $w$ is defined by 
$w \backslash L = \{ x \in V^* \, | \, wx \in L \}$. 
For a word $w=a_1a_2\cdots a_n \in V^*$, $w^R$ is the  palindrome of $w$, that is, 
$(a_1 a_2 \cdots a_n)^R = a_n \cdots a_2 a_1$.

A morphism $h: V^* \rightarrow U^*$ such that $h(a) \in U$ for all $a \in V$ is called a 
\textit{coding}, and it is called a \textit{weak coding} if $h(a) \in U \cup \{ \lambda \}$ for all $a \in V$. 

An {\it involution} over  $V$ is a bijection $\sigma$ : $V \rightarrow V$ 
such that $\sigma=\sigma^{-1}$.  In particular, an involution $\sigma$ over $V$ 
such that $\sigma(a)\not=a$ for all $a\in V$ is called {\it Watson-Crick involution} 
(in molecular computing theory) in a metaphorical sense of DNA complementarity.

In this paper, we fix an involution $\overline{\cdot}$ over $V$ such that $\overline{\overline{a}} = a$
 for $a \in V$ and extend it to $V^*$ in the usual way. 
Note that for all  $x, y \in V^*$, it holds that $(\overline{x})^R = \overline{x^R}$.

\subsection{Hairpin incompletion--A bounded variant of hairpin lengthening}

For the original definitions of the (unbounded) $k$-hairpin completion, the reader 
is referred to precedent papers (for example, \cite{castmit01, tc06, tcs09}).  
A variant of the notion called bounded $k$-hairpin completion and its 
modified operation were introduced and investigated in \cite{bounded09} 
and \cite{kopecki}, respectively,  while a recent paper \cite{hl10} 
introduces and studies an extended version of the hairpin completion, 
called hairpin lengthening. 

In this paper, we are interested in a new variant of both the bounded $k$-hairpin completion and the hairpin lengthening which will be introduced as follows. 

Let $m, k \ge 1$.  For any $w \in V^*$, we define the {\it $m$-bounded $k$-hairpin incompletion} 
of $w$, denoted by $HI_{m,k}(w)$,  as follows:
\begin{align*}
rHI_{m,k} (w) &= \{  w \overline{\gamma}^R \, | \, w = \delta \gamma \alpha \beta \overline{\alpha}^R , \, | \alpha | = k , \, | \gamma | \le m, \, \alpha, \beta, \gamma, \delta \in V^* \}, \\
lHI_{m,k} (w) &= \{  \overline{\gamma}^R w \, | \, w =  \alpha \beta \overline{\alpha}^R \gamma \delta , \, | \alpha | = k , \, | \gamma | \le m, \, \alpha, \beta, \gamma, \delta \in V^* \}, \\
HI_{m,k} (w) &=  rHI_{m,k} (w) \cup lHI_{m,k} (w). 
\end{align*}
where $rHI_{m,k}$  (or $lHI_{m,k}$) is called {\it $m$-bounded  right (or left)  
$k$-hairpin incompletion}. Moreover, $m$-bounded  right (or left)  $k$-hairpin 
incompletion is also called  $m$-bounded  {\it one-sided} $k$-hairpin incompletion.
(See Figure \ref{hi1}, for pictorial  illustration of the operations  $rHI_{m,k}$  and $lHI_{m,k}$.) Thus, from a mathematical viewpoint, we consider the hairpin incompletion operations whose prolongations take place at both ends in a hypothetical (and ideal) 
 molecular biological setting.

{\bf Note.}  For $w \in V^*$ not satisfying the condition to apply the $m$-bounded 
$k$-hairpin incompletion,  here we assume $r\,HI_{m,k}(w)=l\,HI_{m,k}(w)=\{ w\}$.  \\

The iterated version of the $m$-bounded right $k$-hairpin incompletion is defined 
in a usual manner : 
\[
\left\{
\begin{array}{rcl}
r HI^{0}_{m,k} (w) &=& \{ w \}, \\
r HI^{n+1}_{m,k} (w) &=& r HI_{m,k} ( r HI^{n}_{m,k} (w) ) \text{ for } n \ge 0,\\
r HI^{*}_{m,k} (w) &= &\bigcup_{n \ge 0} r HI^{n}_{m,k} (w).
\end{array}
\right.
\]
The "left" counterpart of the iterated version of this operation is defined 
in an obvious and similar manner and is denoted by \ $ l HI^{*}_{m,k} (w)$.  

Further, the iterated version of the  $m$-bounded  $k$-hairpin incompletion 
operation is defined in a similar manner as follows:
\[
\left\{
\begin{array}{rcl}
 HI^{0}_{m,k} (w) &=& \{ w \}, \\
HI^{n+1}_{m,k} (w) &=&  HI_{m,k} ( HI^{n}_{m,k} (w) ) \text{ for } n \ge 0,\\
HI^{*}_{m,k} (w) &= &\bigcup_{n \ge 0} HI^{n}_{m,k} (w).
\end{array}
\right.
\]
Finally, the  iterated version of the $m$-bounded  (right or left)  $k$-hairpin incompletion 
operation is naturally extended to languages  as follows : 
\[
\left\{
\begin{array}{rl}
r HI^{*}_{m,k} (L) &= \bigcup_{w \in L} r HI^{*}_{m,k} (w),\\
l HI^{*}_{m,k} (L) &= \bigcup_{w \in L} l HI^{*}_{m,k} (w), \\[3mm]
\mbox{and \  }HI^{*}_{m,k} (L) &= \bigcup_{w \in L} HI^{*}_{m,k} (w).\\
\end{array}
\right.
\]

\begin{figure}[t]
\centerline{
\includegraphics[scale=0.3]{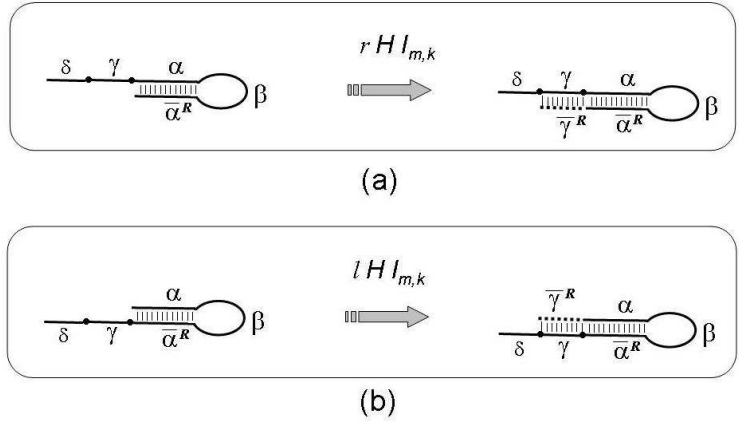}}
\caption{(a) $m$-bounded right $k$-hairpin incompletion operation ;   (b) 
$m$-bounded left $k$-hairpin incompletion operation,  where $|\alpha|=k$ 
and $|\gamma|\leq m$.
}
\label{hi1}
\end{figure}

Note that the bounded hairpin {\it in}completion in this paper is an extension of bounded hairpin completion in the sense that $HI_{m,k}(w)$ is exactly 
the same as $mHC_k(w)$ in \cite{bounded09}  when the prefix (suffix)  $\delta$ of $w$ is empty.  Further, the hairpin lengthening $HL_{k}(w)$ in \cite{hl10} 
is corresponding to the union of all $HI_{m,k}(w)$, where 
$m$ is arbitrary,  in this paper.

\section{Main Results}

\subsection{Non-iterated bounded hairpin incompletion}

As is expected from the definitions, non-iterated bounded hairpin incompletion operation 
behaves as  the bounded hairpin completion operation does.

\begin{thm}
Let $\mathcal{L}$ be a class of languages and $m, k \ge 1$. If $\mathcal{L}$ is closed 
under $gsm$-mappings, $\mathcal{L}$ is also closed under $m$-bounded $k$-hairpin incompletion.
\end{thm}

\begin{proof}
For any $m,k \ge 1$,  consider a generalized sequential machine ({\it gsm}) $g_{m,k}$ 
which adds a suffix (or prefix) $\overline{\gamma}^R$ of length at most $m$ to the 
input word $w$  if  $w$  is of the form $\delta \gamma \alpha \beta \overline{\alpha}^R$ 
(or $\alpha \beta \overline{\alpha}^R \gamma \delta$) with $|\alpha| = k$, $|\gamma| \le m$. 
It is easily shown that this $gsm$ simulates $m$-bounded $k$-hairpin incompletion $HI_{m,k}(w)$.
\end{proof}

Since every trio\footnote{A non-empty family of languages closed under 
$\lambda$-free morphisms, inverse morphisms and intersection with 
regular languages.} is closed under {\it gsm} mapping (\cite{salomaa73}),  the following is 
straightforwardly obtained.
\begin{cor}
Every $trio$ is closed under  $m$-bounded $k$-hairpin incompletion for any $m,k \ge 1$.
\end{cor}

This result extends the corresponding one (Proposition 1) in \cite{bounded09}, 
while it is contrast to the results (Propositions 1 and 2) in \cite{hl10}.  

\subsection{Iterated bounded one-sided hairpin incompletion}

In this section, we consider the closure property of iterated bounded one-sided hairpin incompletion. Especially,  we show that every AFL is closed under this operation. 
To this aim, we start  by preparing some notions required in the proof of the 
main result.  A key idea of the proof is to  construct a certain equivalence relation which is right invariant and of finite index.

First, we consider the iterated $m$-bounded right $k$-hairpin incompletion operation : 
$rHI_{m,k}^*$. 

\begin{de}
Given $m, k \ge 1$ and a word $w \in V^{\ge 2k}$,  we define : 
\begin{align*}
C_{m,k}(w) = \{ (xy, z) \, | \, &xy \in \bigcup_{0 \le i \le m} Inf_{i+k}(w), \, |y|=k, \,  \\
                                                   &w=w_1 xy w_2, \, z \in Suf_{\le k}(w_2) \cap Pref_{\le k}( {\overline{y}}^R ) \}, \\
D_{m,k}(w) = ( C_{m,k}(w)&, \bigcup_{0 \le i \le m} \{ suf_{i+k-1}(w) \} ).
\end{align*}
\label{def-dmk}
\end{de}
We also define a binary relation $\equiv_{D_{m,k}}$ as follows : For $w_1, w_2 \in V^{\ge 2k}$, 
\[
 w_1 \equiv_{D_{m,k}} w_2 \ \mbox{  iff  } \  D_{m,k}(w_1) = D_{m,k}(w_2) . 
\]

Intuitively, a pair $(xy, z)$ in $C_{m,k}(w)$ implies that  it is a candidate of $(\gamma \alpha, 
\overline{\alpha}^R)$ where $\alpha$ and $\gamma$ satisfy the conditions to apply 
$m$-bounded right $k$-hairpin incompletion to $w$, producing a word in $rHI^i_{m,k} (w)$.  

From the definition, it holds that $(\gamma \alpha, \overline{\alpha}^R)$ is in $C_{m,k}(w)$ 
with $|\alpha| = k$ 
if and only if $w{\overline{\gamma}}^R$ 
is in $rHI_{m,k}(w)$.

The binary relation $\equiv_{D_{m,k}}$ is clearly an equivalence relation and of finite index, 
that is, the number of equivalence classes $| V^{\ge 2k} / \equiv_{D_{m,k}}|$ is finite. 
Moreover, the following claim holds.

\begin{cla}
The equivalence relation $\equiv_{D_{m,k}}$ is right invariant, that is,  for $w_1, w_2 \in 
V^{\ge 2k}$,  $w_1 \equiv_{D_{m,k}} w_2$ implies that for any $r \in V^*$, $w_1 r \equiv_{D_{m,k}} w_2 r$.
\label{cla-ri}
\end{cla}

\begin{proof}
We prove it by induction on the length of $r$. If $|r| = 0$,  then the claim trivially holds. 
Assume that  $w_1 \equiv_{D_{m,k}} w_2$  implies that $w_1 r \equiv_{D_{m,k}} w_2 r$  
with  $|r| \geq 1$.    Then, it suffices to show that for any $a \in V$,  $D_{m,k}(w_1 ra) = D_{m,k}(w_2 ra)$. 

We observe  that $D_{m,k}(w_1 ra)$  is constructed from {\it only}  $D_{m,k}(w_1 r)$  as follows:

\begin{align*}
\bigcup_{0 \le i \le m} \{ suf_{i+k-1}&(w_1 ra) \} \\
=& \{  suf_{i+k-2}(w_1 r) \cdot a \, | \,  0 \le i \le m, \, i+k \ge 2, \, |w_1 r| \ge i+k-2 \} \\
  &(\cup \{ \lambda \} \text{ if } k = 1),
\end{align*}
\begin{align*} 
C_{m,k}(w_1 ra) = &\{ (x, \lambda) \, | \, (x, \lambda) \in C_{m,k}(w_1 r) \}  \\
                &\cup \{ (suf_{i+k-1}(w_1 r) \cdot a, \lambda) \, | \, 0 \le i \le m, \, |w_1 r| \ge  i+k-1 \} \\
                 &\cup \{ (xy, za) \, | \, (xy, z) \in C_{m,k}(w_1 r), |y|=k, za \in Pref_{\le k}( {\overline{y}}^R ) \}.
\end{align*}
Note that if $(xy, z) \in C_{m,k}(w_1 r)$, then  $w_1r=w'_1 xy w''_1 z$ for some 
$w', w" \in V^*$,  so that $w_1 r a$ can be rewritten by $w'_1 xy w''_1 za$. Therefore,    
$\{ (xy, za) \, | \, (xy, z) \in C_{m,k}(w_1 r), |y|=k, za \in Pref_{\le k}( {\overline{y}}^R ) \}$ is 
contained in $C_{m,k}(w_1 ra)$.

From the induction hypothesis,  since $D_{m,k}(w_1 r) = D_{m,k}(w_2 r)$, we can construct $D_{m,k}(w_2 ra)$ from \textit{only} $D_{m,k}(w_1 r)$ in the same way. 
Thus, it holds that $D_{m,k}(w_1 ra) = D_{m,k}(w_2 ra)$.
\end{proof}

We first show that the language obtained by applying iterated bounded right hairpin 
incompletion to a singleton is regular. 

\subsubsection*{[Regular grammar $G_w$]}

Let's consider the equivalence classes :
\[
V^{\ge 2k} / \equiv_{D_{m,k}} = \{ [w_1], [w_2], \dots, [w_t] \, | \, w_i \in V^{\ge 2k}, 1 \le i \le t \},
\]
where $w_i$ is the representative of $[w_i]$. For $w \in V^{\ge 2k}$, the regular grammar 
$G_w = (N, V, P, S)$ is constructed as follows : 
\begin{align*}
N = &\{ S \} \cup \{ D_i \,| \, 1 \le i \le t \}, \\
P = &\{ S \rightarrow wD_i \, | \, w \in V^{\ge 2k}, w \equiv_{D_{m,k}} w_i \} \\
&\cup \{ D_i \rightarrow rD_j \, | \, (\gamma \alpha, \overline{\alpha}^R) \in C_{m,k}(w_i), |\alpha|=k,  \\
       & \ \ \hspace*{22mm} r = \overline{\gamma}^R , \ w_i r \equiv_{D_{m,k}} w_j, \ 
 1 \le i, j \le t \ \} \\
&\cup \{ D_i \rightarrow \lambda \, | \, 1 \le i \le t \}.
\end{align*}
We need the following two claims.
\begin{cla}
Let $w$ be in $V^{\ge 2k}$, and $D_i, D_j \in N$. Then, for $n \ge 0$,  if 
a derivation of $G_w$ is of the form $wD_i \Rightarrow^n wrD_j$  
for some  $r \in V^*$, then it holds $wr \equiv_{D_{m,k}} w_j$.
\label{cla-sf}
\end{cla}

\begin{proof}
The proof is by induction on $n$. If $n=0$, then $i=j$ and from the 
manner of constructing $P$,  it holds $w \equiv_{D_{m,k}} w_j$,  thus, the claim holds.  
Assume that the claim holds for $n>0$ and consider a derivation 
 of the form $wD_i \Rightarrow wrD_j \Rightarrow^n wrr'D_h$ for 
some $D_h \in N$, $r' \in V^*$. 
 From the assumption and the form of $P$, it holds that $wr \equiv_{D_{m,k}} w_j$ and 
$w_j r' \equiv_{D_{m,k}} w_h$.  By Claim \ref{cla-ri}, we obtain that $wrr' \equiv_{D_{m,k}} w_j r' \equiv_{D_{m,k}} w_h$.
\end{proof}

\begin{cla}
For $n \ge 0$ and $r \in V^*$, there exists a derivation  of $G_w$ of the form 
$S \Rightarrow wD_i \Rightarrow^n wrD_j \Rightarrow wr$ if and only if $wr$ is in $rHI^n_{m,k}(w)$.
\label{cla-sf2}
\end{cla}

\begin{proof}
The proof is by induction on $n$. If $n=0$, it obviously holds that $S \Rightarrow wD_i  \Rightarrow w$ 
if and only if $w$ is in $rHI^0_{m,k}(w)$. Assume that the claim holds for $n$ and consider the case for $n+1$. \\

(If Part) Let $wr' \in rHI^{n+1}_{m,k}(w)$. Then there exists $r, \gamma \in V^*$ such that 
$wr' = wr{\overline{\gamma}}^R \in rHI_{m,k}(wr)$ with $wr \in rHI^n_{m,k}(w)$. From the 
definition of $C_{m,k}$, $(\gamma \alpha, \overline{\alpha}^R)$ is in $C_{m,k}(wr)$ with 
$|\alpha| = k$. From the induction hypothesis and Claim \ref{cla-sf}, there exists a derivation : 
$S \Rightarrow wD_i \Rightarrow^n wrD_j$ with $wr \equiv_{D_{m,k}} w_j$. Since 
$(\gamma \alpha, \overline{\alpha}^R)$ is in $C_{m,k}(wr) = C_{m,k}(w_j)$, there exists 
the derivation $S \Rightarrow wD_i \Rightarrow^n wrD_j \Rightarrow wr{\overline{\gamma}}^R D_h \Rightarrow  wr{\overline{\gamma}}^R = wr'$ for some $D_h \in N$. \\

(Only If Part) If there exists the derivation $S \Rightarrow wD_i \Rightarrow^n wrD_j 
\Rightarrow wr{\overline{\gamma}}^R D_h$ $\Rightarrow  wr{\overline{\gamma}}^R$ for some
 $D_h \in N$, it holds that $wr \equiv_{D_{m,k}} w_j$ from Claim \ref{cla-sf}. Moreover, from 
the form of $P$, there exists $(\gamma \alpha, \overline{\alpha}^R) \in C_{m,k}(w_j) = C_{m,k}(wr)$. 
Hence, $wr{\overline{\gamma}}^R$ is in $rHI_{m,k}(wr)$. From the induction hypothesis, 
$wr \in rHI^n_{m,k}(w)$ so that $wr{\overline{\gamma}}^R \in rHI^{n+1}_{m,k}(w)$.
\end{proof}

It follows from the claim  that the language obtained by applying iterated bounded right hairpin incompletion to a singleton is regular. 

\begin{lem}
For any word $w \in V^*$ and $m, k \ge 1$,   a language $rHI^{*}_{m,k} (w)$ is regular.
\label{lem-single}
\end{lem}

\begin{proof}
In the case of $w \in V^* - V^{\ge 2k}$,  from the definition,  $rHI^{*}_{m,k} (w) = \{w\}$ 
is regular.  For $w \in V^{\ge 2k}$ it follows from Claim 
\ref{cla-sf2}  that there exists a derivation of $G_w$ which 
derives a terminal string $w'$ if and only if $w' \in rHI^{*}_{m,k}(w)$. Thus, we have 
that $L(G_w) =  rHI^{*}_{m,k} (w)$ which is regular.
\end{proof}

In order to show more general results, we need to prove the claims regarding the  
language  $rHI^{*}_{m,k} (w)$.

\begin{cla}
For $w_1, w_2 \in V^{\ge 2k}$ and $n \ge 0$, if $w_1 \equiv_{D_{m,k}} w_2$ then there 
exists a finite language $F \subseteq V^*$ such that $rHI^{n}_{m,k} (w_1) = w_1 F$ 
and $rHI^{n}_{m,k} (w_2) = w_2 F$. 
\label{cla-fin}
\end{cla}

\begin{proof}
The proof is by induction on $n$. If $n=0$, it obviously holds that
$rHI^{0}_{m,k} (w_1) = w_1 F$ and $rHI^{0}_{m,k} (w_2) = w_2 F$, 
where $F=\{\lambda\}$.  We assume that the claim holds for up to $n$. 
Let $rHI^{n}_{m,k} (w_1) = w_1 F$ and $rHI^{n}_{m,k} (w_2) = w_2 F$ for some finite 
language $F$. For any  $r \in F$, it holds that $w_1 r \equiv_{D_{m,k}} w_2 r$ from 
Claim \ref{cla-ri}. Hence, from the induction hypothesis, there exists a finite 
language $F_r$ such that \[ rHI_{m,k} (w_1 r) = w_1 r F_r \text{ and } rHI_{m,k} (w_2 r) = w_2 r F_r. \] 
Therefore, it holds that 
\begin{align*}
rHI^{n+1}_{m,k} (w_1) &= rHI_{m,k} (w_1 F) = \bigcup_{r \in F} w_1 r  F_r = w_1 \bigcup_{r \in F} r  F_r = w_1 F', \\
rHI^{n+1}_{m,k} (w_2) &= rHI_{m,k} (w_2 F) = \bigcup_{r \in F} w_2 r  F_r = w_2 \bigcup_{r \in F} r  F_r = w_2 F', 
\end{align*}
where $F' = \bigcup_{r \in F} r  F_r$.
\end{proof}

\begin{cla}
For $w_1, w_2 \in V^{\ge 2k}$, if $w_1 \equiv_{D_{m,k}} w_2$ then there exists a 
regular language $R \subseteq V^*$ such that $rHI^{*}_{m,k} (w_1) = w_1 R$ and $rHI^{*}_{m,k} (w_2) = w_2 R$.
\label{cla-reg}
\end{cla}

\begin{proof}
From Claim \ref{cla-fin}, if  $w_1 \equiv_{D_{m,k}} w_2$,  then there exists a sequence 
of finite languages : $F_0, F_1, F_2, \cdots, $\ where $F_n \subseteq V^* (n \ge 0)$, 
with the property  that for $n \ge 0$,  $rHI^{n}_{m,k} (w_1) = w_1 F_n$ and 
$rHI^{n}_{m,k} (w_2) = w_2 F_n$. Then it holds that
\begin{align*}
rHI^{*}_{m,k} (w_1) &= \bigcup_{n \ge 0} rHI^{n}_{m,k} (w_1) = \bigcup_{n \ge 0} w_1 F_n = w_1 \bigcup_{n \ge 0} F_n, \\
rHI^{*}_{m,k} (w_2) &= \bigcup_{n \ge 0} rHI^{n}_{m,k} (w_2) = \bigcup_{n \ge 0} w_2 F_n = w_2 \bigcup_{n \ge 0} F_n. 
\end{align*}
Let  $R = \bigcup_{n \ge 0} F_n$. Then, we obtain $rHI^{*}_{m,k} (w_1) = w_1 R$ and 
$rHI^{*}_{m,k} (w_2) = w_2 R$. Recall that $w_1 R$ and $w_2 R$ are regular from 
Lemma \ref{lem-single}. The class of regular languages is closed under left derivative, 
so that $R$ is also regular.
\end{proof}

We are now in a position to show the main theorem in this section.  It is shown that iterated bounded one-sided hairpin incompletion can be simulated by several 
basic language operations, which leads to the following theorem.

\begin{thm}
Let $\mathcal{L}$ be a class of languages and $m, k \ge 1$. If $\mathcal{L}$ is closed 
under intersection with regular languages, concatenation with regular languages and 
finite union, then $\mathcal{L}$ is also closed under iterated  $m$-bounded right (left) $k$-hairpin incompletion.
\end{thm}

\begin{proof}
Let $L \in \mathcal{L}$ be the language over $V$. We can write $L = L_1 \cup L_2$ where
\begin{align*}
L_1 &= L \cap V^{2k} \cdot V^* = \{ w \in L \, | \, |w| \ge 2k \}, \\
L_2 &= L \cap \bigcup_{0 \le n \le 2k-1}V^n = \{ w \in L \, | \, |w| < 2k \}. 
\end{align*}
Note that $rHI^*_{m,k}(L) = rHI^*_{m,k}(L_1) \cup rHI^*_{m,k}(L_2) = rHI^*_{m,k}(L_1) \cup L_2$.
Since the number of the elements in $L_1 / \equiv_{D_{m,k}}$ is finite from the definition of 
$\equiv_{D_{m,k}}$, we can set  
$L_1 / \equiv_{D_{m,k}} = \{ [w_1], [w_2], \dots, [w_s] \, | \, w_i \in L_1 \text{ for } 1 \le i \le s \}$ 
for some $s \ge 0$. From the way of construction of $D_{m,k}(w_i)$, it holds that for $1\le i \le s$, 
\[ [w_i] =  L_1 \cap ( \bigcap_{(xy,z) \in C_{m,k}(w_i)} V^* xy V^* z ) \cap ( \bigcap_{0 \le j \le m} V^* \cdot suf_{j+k-1}(w_i) ), \] 
For $1 \le i \le s$,  since all words in $[w_i]$ are equivalent, it follows from Claim \ref{cla-reg} that 
there exists regular language $R_i$ such that $rHI^*_{m,k}([w_i]) = [w_i] R_i$. Moreover, it holds 
that $rHI^*_{m,k}(L_1) = \bigcup_{1 \le i \le s}  [w_i] R_i$. Thus, $rHI^*_{m,k}(L)$ can be 
constructed from $L$ by intersection with regular languages, concatenation with regular 
languages and finite union, which completes the proof. 
\end{proof}

As a corollary, we immediately obtain the following. 

\begin{cor}
Every $AFL$ is closed under iterated  $m$-bounded right (left) $k$-hairpin incompletion for any $m,k \ge 1$.
\end{cor}

It is known in \cite{PRS98}  that there exists  {\it no universal} regular grammar 
$G_u(x)=(V, \Sigma, P, x)$ with the property that for any regular grammar $G$, 
there exists a coding $w_G$ of $G$ such that $L(G)=L(G_u(w_G))$.  This can be strengthened 
in the form that no morphism $h$ can help to satisfy the equation $L(G)=h(L(G_u(w_G)))$. 
    
In this context, the next lemma shows that the bounded hairpin incompletion 
operation  can play a role of the universal-like grammar  for all regular languages. 

\begin{lem}
A language $L \subseteq V^*$ is regular if and only if there exists a word $w \in (V')^*$ 
and a weak coding $h: V' \rightarrow V$ such that $L = h(rHI^{*}_{1,1} (w) \cap {(V' - \{ \# \} )}^* V'' )$, 
where $\# \in V'$ and $V'' \subseteq V'$.
\label{renma2}
\end{lem}

\begin{proof}
(If Part)  This  clearly holds,  because the class of the regular languages is closed 
under iterated bounded right hairpin incompletion, intersection and weak codings.

(Only If Part) For a regular grammar $G = (N, V, P, S)$, we construct $V'$,$V''$, 
$w \in V$ and $h: V' \rightarrow V$ as follows:
\begin{itemize}
\item $V' = \{ [a, X] \, | \, a \in V, X \in N \cup \{ \lambda \} \} \cup  \{ \overline{[a, X]} \, | \, a \in V, X \in N \} \cup  \{ \#, \overline{ \# } \}$,
\item $V'' =  \{ [a, \lambda] \, | \, a \in V \}$,
\item $\displaystyle w = ( \prod_{X_i \rightarrow aX_j \in P, \, b \in V}  \overline{ \# } \, \overline{[a, X_j]} \, \overline{[b, X_i]} ) \cdot [\lambda, S]$,
\item $h(A) = a$ for $A = [a, X] \in \{ [a, X] \, | \, a \in V, X \in N \cup \{ \lambda \} \}$, $h(A) = \lambda$ otherwise.
\end{itemize}
Note that for any $n \ge 0$ and $w' =  \delta \gamma \alpha \beta \overline{\alpha}^R  \in rHI^{n}_{1,1} (w) \cap {(V' - \{ \# \} )^*}$ with $|\alpha| = |\gamma| = 0$,  if $w \overline{\gamma}^R \in rHI^{n+1}_{1,1} (w) \cap {(V' - \{ \# \} )^*}$,  then $\gamma$ is the symbol just right of $\overline{ \# }$. Then, from the way of construction of $w$, it holds that there exists a derivation of $G$, 
\[S \Rightarrow a_1 X_1 \Rightarrow a_1 a_2 X_2 \Rightarrow \dots \Rightarrow a_1 a_2 \dots a_{n-1} X_{n-1} \Rightarrow a_1 a_2 \dots a_{n-1} a_n, \]
if and only if
\[ w' = ( \prod_{X_i \rightarrow aX_j \in P, \, b \in V}  \overline{ \# } \, \overline{[a, X_j]} \, \overline{[b, X_i]} )  [\lambda, S] [a_1, X_1] [a_2, X_2] \dots [a_{n-1}, X_{n-1}] [a_n, \lambda] \]
is in $rHI^{n}_{1,1} (w) \cap {(V' - \{ \# \} )^*}$, which can be shown by induction on $n$.  By 
applying $h$, we obtain $L(G) = h(rHI^{*}_{1,1} (w) \cap {(V' - \{ \# \} )}^* V'' )$.
\end{proof}

We note that Theorem 3 in \cite{hl10} proves the {\it only if} part of this 
 lemma for the iterated (unbounded) hairpin lengthening. Thus, 
Lemma \ref{renma2} complements the result for the case of bounded hairpin 
lengthening.

\subsection{Iterated bounded hairpin incompletion}

In this section, we consider the closure property of iterated bounded hairpin incompletion. For the (unbounded) hairpin lengthening operation, the paper \cite{hl10} has proved that the family of context-free languages is closed 
under iterated hairpin lengthening in Theorem 4. We will show that the 
result also holds for the case of iterated bounded hairpin lengthening, 
in a more general setting of AFL-like formulation.

The proof is based on the similar idea to the previous  section and Claim \ref{cla-ri}, \ref{cla-sf}, \ref{cla-sf2} are corresponding to Claim \ref{cla-rili}, \ref{cla-sfcf}, \ref{cla-sfcf2} (below), respectively.

In order to consider both-sided hairpin incompletion, we modify the equivalence relation.

\begin{de}
For $m, k \ge 1$ and the word $w \in V^{\ge 2k}$, $C'_{m,k}(w)$, $D'_{m,k}(w)$ and $E_{m,k}(w)$ are defined by
\begin{align*}
C'_{m,k}(w) &= \{ (z, yx) \, | \, yx \in \bigcup_{0 \le i \le m} Inf_{i+k}(w), \, |y|=k, \\
                      & \hspace*{18mm}  w=w_1 yx w_2, \, z \in Pref_{\le k}(w_1) \cap Suf_{\le k}( {\overline{y}}^R ) \}, \\
D'_{m,k}(w) &= ( C'_{m,k}(w), \bigcup_{0 \le i \le m} \{ pref_{i+k-1}(w) \} ), \\
E_{m,k}(w) &=  < D_{m,k}(w), D'_{m,k}(w) >.
\end{align*}
, where $D_{m,k}(w)$ is the relation defined in Definition \ref{def-dmk}.

The binary relation $\equiv_{E_{m,k}}$ is defined as $w_1 \equiv_{E_{m,k}} w_2$ if $E_{m,k}(w_1) = E_{m,k}(w_2)$ for $w_1, w_2 \in V^{\ge 2k}$. 
\end{de}
The binary relation $\equiv_{E_{m,k}}$ is clearly an equivalence relation and of finite index. Note that $D_{m,k}$ and $D'_{m,k}$ are symmetrically defined.

We show that the  equivalence relation $\equiv_{E_{m,k}}$ is right invariant and left invariant.

\begin{cla}
The equivalence relation $\equiv_{E_{m,k}}$ is right invariant and left invariant, that is,  for $w_1, w_2 \in V^{\ge 2k}$, if $w_1 \equiv_{E_{m,k}} w_2$ then for any $r, l \in V^*$, $w_1 r \equiv_{D_{m,k}} w_2 r$ and $l w_1 \equiv_{E_{m,k}} l w_2$ holds.
\label{cla-rili}
\end{cla}

\begin{proof}
We firstly show that for $r \in V^*$, $w_1 r \equiv_{D_{m,k}} w_2 r$. The proof is by induction on the length of $r$. If $|r|=0$, it clearly holds. We assume that the claim holds for $n$, i.e.,  $w_1 r \equiv_{E_{m,k}} w_2 r$ with $|r|=n$. Let $a$ be a symbol in $V$. 

[Proof of $D_{m,k}(w_1 ra) = D_{m,k}(w_2 ra)$] 
It can be shown by the same way as Claim \ref{cla-ri}. 

[Proof of $D'_{m,k}(w_1 ra) = D'_{m,k}(w_2 ra)$] 
We construct $D'_{m,k}(w_1 ra)$ from \textit{only} $E_{m,k}(w_1 r)$ as follows:
\begin{align*}
\bigcup_{0 \le i \le m} \{ pref_{i+k-1}(w_1 ra) \} = &\{ pref_{i+k-1}(w_1 r) \, | \, 0 \le i \le m, \, |w_1 r| \ge i+k-1  \} \\
(&\cup \{ w_1 ra \} \text{ if }  |w_1 r| < m+k-1), 
\end{align*}
\begin{align*}
C'_{m,k}(w_1 ra) =&  \, C'_{m,k}(w_1 r) \\
&\cup \, \{ ( \lambda, suf_{i+k-1}(w_1r) \cdot a) \, | \, 0 \le i \le m, \, |w_1 r| \ge i+k-1 \} \\
&\cup \, \{ (z, suf_{i+k-1}(w_1r) \cdot a) \, | \, 0 \le i \le m, \, |w_1 r| \ge |z| + i+k-1, \\
&\ \hspace{25mm} z \in Pref_{\le k}(w_1 r) \cap Suf_{\le k}( \overline{suf_{i+k-1}(w_1r) \cdot a}^R  ) \}.
\end{align*}
Note that for $0 \le i \le m$, $z \in Pref_{\le k}(w_1 r) \cap Suf_{\le k}( \overline{suf_{i+k-1}(w_1 r) \cdot a}^R  )$ with $|w_1 r| \ge |z| + i+k-1$ and some $z' \in V^*$,  $w_1 r a$ can be represented as $w_1 r a = z \cdot z' \cdot suf_{i+k-1}(w_1r) \cdot a$. Hence, $(z, suf_{i+k-1}(w_1r) \cdot a)$ is in $C'_{m,k}(w_1 ra)$.

Since $E_{m,k}(w_1 r) = E_{m,k}(w_2 r)$, we can construct $D'_{m,k}(w_2 ra)$ from \textit{only} $E_{m,k}(w_1 r)$ in the same way. Therefore, it holds that $D'_{m,k}(w_1 ra) = D'_{m,k}(w_2 ra)$.
From $D_{m,k}(w_1 ra) = D_{m,k}(w_2 ra)$ and $D'_{m,k}(w_1 ra) = D'_{m,k}(w_2 ra)$, we eventually  get $w_1 ra\equiv_{E_{m,k}} w_2 ra$. 

For the left invariance of $\equiv_{E_{m,k}}$, we can show in the symmetrical manner.
\end{proof}

\subsubsection*{[Linear grammar $G_L$]}

For the proof of Theorem \ref{ibhi} (below) regarding $m$-bounded $k$-hairpin incompletion, we need to construct a linear grammar. For $L \subseteq V^*$, let $L / \equiv_{E_{m,k}} = \{ A_1, A_2, \dots, A_u  \}$ for some $u \ge 1$ and $V^* / \equiv_{E_{m,k}} = \{ [w_1], [w_2], \dots, [w_s]  \}$ for some $s \ge 1$, where $w_i$  is the representative of  $[w_i]$. A linear grammar $G_L = ( N, T, P, S )$ is constructed as follows:
\begin{align*}
N =& \{ S \} \cup \{ E_i \, | \, 0 \le i \le s \}, \\
T =& V \cup \{ a_i \, | \, 0 \le i \le u \} \cup \{ \$ \}, \\
P =& \{ S \rightarrow E_i a_j \, | \, \text{For any } w \in A_j, w  \equiv_{E_{m,k}} w_i  \} \\
   &\cup \{ E_i \rightarrow r E_j \, | \, (\gamma \alpha, \overline{\alpha}^R) \in C_{m,k}(w_i), |\alpha|=k, r = {\overline{\gamma}}^R ,  w_i r  \equiv_{E_{m,k}} w_j  \} \\
   &\cup \{ E_i \rightarrow E_j l \, | \, (\overline{\alpha}^R,  \alpha \gamma ) \in C'_{m,k}(w_i), |\alpha|=k, l = {\overline{\gamma}}^R,  l w_i  \equiv_{E_{m,k}} w_j  \} \\
   &\cup \{ E_i \rightarrow \$ \, | \,  0 \le i \le s \}.
\end{align*}

We set $R_P = \{ r \, | \, E_i \rightarrow r E_j \in P  \} \cup \{ \lambda \}$ and $L_P = \{ l \, | \, E_i \rightarrow E_j l \in P  \} \cup \{ \lambda \}$.
\begin{cla}
Let $0 \le p \le u$ and $E_i, E_j \in N$.   For $n \ge 0$, if a derivation of $G_L$ is of the form $E_i a_p \Rightarrow^n r_1 \dots r_n E_j l_n \dots l_1 a_p$, then for any $w \in A_p$, it holds that $l_n \dots l_1 w r_1 \dots r_n \equiv_{E_{m,k}}  w_j$, where for each $1 \le h \le n$, $r_h \in R_P$, $l_h \in L_P$, one of $r_h$ and $l_h$ is $\lambda$ and the other is not $\lambda$.
\label{cla-sfcf}
\end{cla}

\begin{proof}
The proof is by induction on $n$. If $n=0$, then $i=j$ and from the manner of constructing $P$, for any $w \in A_p$, it holds that $w \equiv_{E_{m,k}}  w_j$, thus the claim holds. 
Assume that the claim holds for $n>0$ and consider a derivation  
of the form 
\begin{align*}
&E_i a_p \Rightarrow r' E_j  a_p\Rightarrow^n r' r_1 \dots r_n E_h l_n \dots l_1  a_p  \\
&( \, E_i a_p \Rightarrow E_j l' a_p\Rightarrow^n r_1 \dots r_n E_h l_n \dots l_1 l' a_p \, )
\end{align*}
for some $E_h \in N$, $r' \in R_P$ $(l' \in L_P)$. 
From the assumption and the form of $P$, for any $w \in A_p$, it holds that $wr' \equiv_{E_{m,k}} w_j$ $(l' w \equiv_{E_{m,k}} w_j)$ and $l_n \dots l_1 w_j r_1 \dots r_n \equiv_{E_{m,k}} w_h$.  
By Claim \ref{cla-rili}, we obtain that 
\begin{align*}
&l_n \dots l_1wr' r_1 \dots r_n \equiv_{E_{m,k}} l_n \dots l_1 w_j r_1 \dots r_n \equiv_{E_{m,k}} w_h  \\
&( \, l_n \dots l_1 l' wr_1 \dots r_n \equiv_{E_{m,k}} l_n \dots l_1 w_j r_1 \dots r_n \equiv_{E_{m,k}} w_h \, ).
\end{align*}
\end{proof}

\begin{cla}
A word $r_1 \dots r_n \$ l_n \dots l_1 a_i$  is generated by $G_L$ if and only if for any $w \in A_i$, $l_n \dots l_1 w r_1 \dots r_n$ is in $HI^n_{m,k}(L)$, where for each $1 \le h \le n$, $r_h \in R_P$, $l_h \in L_P$, one of $r_h$ and $l_h$ is $\lambda$ and the other is not $\lambda$.
\label{cla-sfcf2}
\end{cla}

\begin{proof}
The proof is by induction on $n$. If $n=0$, it obviously holds that $S \Rightarrow E_i a_j  \Rightarrow \$ a_j$ 
if and only if for any $w \in A_j$, $w$ is in $HI^0_{m,k}(L)$. 
Assume that the claim holds for $n$ and consider the case for $n+1$. \\

(If Part) Let $l_{n+1} l_n \dots l_1 w r_1 \dots r_n r_{n+1} \in HI^{n+1}_{m,k}(w)$, where for each $1 \le h \le n+1$, $r_h \in R_P$, $l_h \in L_P$, one of $r_h$ and $l_h$ is $\lambda$ and the other is not $\lambda$. 
From the definition of $C_{m,k}$ and $C'_{m,k}$, either $( {\overline{r_{n+1}}}^R \cdot \alpha, \overline{\alpha}^R)$ is in $C_{m,k}(l_n \dots l_1 w r_1 \dots r_n)$ or $(\overline{\alpha}^R, \alpha \cdot {\overline{l_{n+1}}}^R)$ is in $C'_{m,k}(l_n \dots l_1 w r_1 \dots r_n)$ with $|\alpha| = k$. 
From the induction hypothesis and Claim \ref{cla-sfcf}, there exists a derivation : 
\[S \Rightarrow E_i a_p \Rightarrow^n r_1 \dots r_n E_j l_n \dots l_1 a_p \]
with $l_n \dots l_1 wr_1 \dots r_n \equiv_{E_{m,k}} w_j$. Therefore, it holds that either $( {\overline{r_{n+1}}}^R \cdot \alpha, \overline{\alpha}^R) \in C_{m,k}(w_j)$ or $(\overline{\alpha}^R, \alpha \cdot {\overline{l_{n+1}}}^R) \in C'_{m,k}(w_j)$, from which there exists the derivation either \\[-8mm]
\begin{align*}
S &\Rightarrow E_i a_p \Rightarrow^n r_1 \dots r_n  E_j l_n \dots l_1 a_p \Rightarrow r_1 \dots r_n r_{n+1} E_h l_n \dots l_1 a_p \\
   &\Rightarrow r_1 \dots r_n r_{n+1} \$ l_n \dots l_1 a_p\\[-8mm]
\end{align*}
or  \\[-10mm]
\begin{align*}
S &\Rightarrow E_i a_p \Rightarrow^n r_1 \dots r_n  E_j l_n \dots l_1 a_p \Rightarrow r_1 \dots r_n E_h l_{n+1} l_n \dots l_1 a_p \\
   &\Rightarrow r_1 \dots r_n \$ l_{n+1} l_n \dots l_1 a_p\\[-7mm]
\end{align*}
for some $E_h \in N$. \\[1mm]
(Only If Part) Consider the case where there exists a derivation $S \Rightarrow E_i a_p \Rightarrow^n r_1 \dots r_n  E_j l_n \dots l_1 a_p \Rightarrow r_1 \dots r_n r_{n+1} E_h  l_n \dots l_1 a_p \Rightarrow r_1 \dots r_n  r_{n+1} \$ l_n \dots l_1 a_p$ for some $E_h \in N$. Then, it holds that for any $w \in A_p$, $l_n \dots l_1 wr_1 \dots r_n \equiv_{E_{m,k}} w_j$ from Claim \ref{cla-sfcf}. Moreover, from 
the way of construction of $P$, there exists $(\overline{r_{n+1}}^R \cdot \alpha, \overline{\alpha}^R) \in C_{m,k}(w_j) = C_{m,k}(l_n \dots l_1 wr_1 \dots r_n)$. 
Hence, $l_n \dots l_1 wr_1 \dots r_n r_{n+1}$ is in $HI_{m,k}(l_n \dots l_1 wr_1 \dots r_n)$. From the induction hypothesis, 
$l_n \dots l_1 wr_1 \dots r_n \in HI^n_{m,k}(w)$ so that $l_n \dots l_1 wr_1 \dots r_n r_{n+1} \in HI^{n+1}_{m,k}(w)$.

For the other case, there exists a derivation $S \Rightarrow E_i a_p \Rightarrow^n r_1 \dots r_n  E_j l_n \dots l_1 a_p$ $\Rightarrow r_1 \dots r_n E_h l_{n+1}  l_n \dots l_1 a_p \Rightarrow r_1 \dots r_n \$ l_{n+1} l_n \dots l_1 a_p$ for some $E_h \in N$. Then we can show in a similar way that for any $w \in A_p$, $l_{n+1} l_n \dots l_1 wr_1 \dots r_n \in HI^{n+1}_{m,k}(w)$.
\end{proof}

In order to prove the next result, we need a language operation called \textit{circular permutation} $cp$ which maps every word in the set of all its circular permutations and every language in the set of all circular permutations of its words. The proof is due to an idea similar to the one in \cite{bounded09}.

\begin{thm}
Let $\mathcal{L}$ be a class of languages which includes all linear languages and let $m, k \ge 1$. If $\mathcal{L}$ is closed under circular permutation, left derivative and substitution, then $\mathcal{L}$ is also closed under iterated  $m$-bounded $k$-hairpin incompletion.
\label{ibhi}
\end{thm}

\begin{proof}
Recall the construction of the linear grammar $G_L$. Let $L$ be in $\mathcal{L}$ and $f$ be a substitution over $T$ defined by $f(a_i) = A_i$ for $\{ a_i \, | \, 0 \le i \le u \}$ and $f(a)= \{a \}$ otherwise. From Claim \ref{cla-sfcf2}, it holds that 
\begin{align*}
L_G = \{ r_1 \dots r_n \$ l_n \dots l_1 a_i \, | \, &a_i \in T, \, 1\le j \le n, \, r_j \in R_P, \, l_j \in L_P,  \\
&\text{for any } w \in A_i, \, l_n \dots l_1 w r_1 \dots r_n \in HI^*_{m,k}(L) \},
\end{align*}
where $L_G = L(G_L)$. Hence, it is easily seen that $HI^*_{m,k}(L) = f( \$ \backslash  cp(L_G))$.
\end{proof}
Since the family of context-free languages meets all of preconditions in Theorem \ref{ibhi}, the following corollary holds.
\begin{cor}
The family of context-free languages is closed under iterated  $m$-bounded $k$-hairpin incompletion for any $m,k \ge 1$.
\end{cor}

\section{Concluding Remarks}

In many works on DNA-based computing and the related areas, DNA hairpin structures have numerous applications to develop novel computing mechanisms in molecular computing. Among others, these molecules of  hairpin formation called Whiplash PCR  have been successfully employed as the basic feature of new computational models to solve an instance of the 3-SAT problem  (\cite{Science}),  to execute (and simulate) state transition systems (\cite{Sakamoto99BS}), to explore the feasibility of  parallel computing for solving DHPP (\cite{kK06a}), and so forth.
On the other hand,  different types of hairpin and hairpin-free languages are defined in \cite{hair}  and more recently in \cite{dlt05}, where they are studied from a language theoretical point of view.

We have proposed a new variant of hairpin completion called hairpin incompletion, and investigated its closure properties of the language 
families. The hairpin incompletion is in fact a bounded variant of 
the hairpin lengthening in \cite{hl10}  where not only closure properties 
of language families but also the algorithmic aspects of the hairpin lengthening operations are investigated.  The hairpin incompletion 
is also an extended version of the bounded hairpin completion recently studied in \cite{bounded09} that has been more recently followed up by slightly modified operations in \cite{kopecki} where two open problems from \cite{bounded09} have been solved.   

We have shown that every AFL is closed under the iterated one-sided hairpin incompletion, and therefore, the family of regular languages is closed under the operation. Further, it has been shown that the family of context-free languages is closed under the iterated hairpin incompletion. These complement some of the corresponding results for (unbounded) hairpin 
lengthening operations in \cite{hl10}.  Moreover, since the hairpin incompletion nicely models a bio-molecular technique (Whiplash 
PCR),  the obtained results in this paper may provide new insight into the computational analysis of the experimental technique. 

It remains as an interesting open problem if the family of regular languages 
is closed under iterated hairpin incompletion.

\end{document}